\documentclass[twocolumn,10pt]{asme2ej}

\usepackage{graphicx} 
\usepackage{hyperref}   
\hypersetup{
	colorlinks=true,
	linkcolor=blue,
	citecolor=blue,
	urlcolor=blue,
}
\usepackage[square,numbers]{natbib}

\usepackage{amssymb}  

\usepackage{amsthm}
\usepackage{color,soul}
\usepackage{flushend}
\usepackage{mathtools}
\usepackage{xcolor}
\usepackage{array,multirow}
\usepackage[ruled,linesnumbered]{algorithm2e}

\usepackage{subfig}

\newtheorem{thmm}{Theorem}
\newtheorem{propp}{Proposition}
\newtheorem{remm}{Remark}
\newtheorem{assm}{Assumption}
\newtheorem{deff}{Definition}

\newtheorem{Cor}{Corollary}

\title{On Distinguishability of Anomalies as Physical Faults or Actuation Cyberattacks}

\author{Tanushree Roy\thanks{Address all correspondence related to this paper to Tanushree Roy.}
    \affiliation{
 Department
of Mechanical Engineering \\Texas Tech University\\ 2500 Broadway Lubbock, Texas 79409\\ E-mail: tanushree.roy@ttu.edu. 
    }	
}

\author{Satadru Dey
    \affiliation{Department
of Mechanical Engineering\\ The Pennsylvania State University\\ University Park, Pennsylvania 16802, USA\\E-mail: skd5685@psu.edu.
    }
}

\begin{document}

\maketitle    

\begin{abstract}
{\it Increased automation has created an impetus to integrate infrastructure with wide-spread connectivity in order to improve efficiency, sustainability, autonomy, and security. Nonetheless, this reliance on connectivity and the inevitability of complexity in this system increases the vulnerabilities to physical faults or degradation and external cyber-threats.  However, strategies to counteract faults and cyberattacks would be widely different and thus it is vital to not only detect but also to identify the nature of the anomaly that is present in these systems. In this work, we propose a mathematical framework to distinguish between physical faults and cyberattack using a sliding mode based  unknown input observer. Finally, we present simulation case studies to distinguish between physical faults and cyberattacks using the proposed Distinguishability metric and criterion. The simulation results show that the proposed framework successfully distinguishes between faults and cyberattacks.
}
\end{abstract}



\section{Introduction}

The growing need for efficiency, coordination, precision, and autonomy has led to the integration of cybernetic components with physical infrastructure through Information and Communication Technologies (ICT). Such physical systems with embedded networks of sensors, actuators, controllers are commonly described as Cyber-physical systems (CPS). Currently, such CPS has garnered a lot of interest in the areas of smart grid \cite{he2016cyber}, manufacturing \cite{zheng2018smart}, mobility \cite{rawat2015towards} and many others. 
Thus, for reliable operation of these safety-critical systems, ensuring safety and security of these systems against faults and cyberattacks has become obligatory. 

\subsection{Motivation}

The impact of fault and cyberattack on CPS may be disparate \cite{teixeira2015secure}. On one hand, faults may arise due to natural degradation of system components or physical abuse. On the other, cyberattack is specifically crafted by an adversary to drive system towards unintended states while evading detection by the system administrator. The wide-range of possibilities for physical faults and cyberattacks also make it challenging to distinguish between them from system measurements. Particularly, faults can be incipient or rapidly evolving leading to runway effects \cite{safaeipour2021survey}. In contrast, some cyberattacks can be passive (such as eavesdropping attack) or stealthy or can deny services from the system altogether. Additionally, the adversary can also design cyberattacks such that it can mimic behavior of faults in systems \cite{rahman2014cyber} or coordinate series of multiple faults in the systems \cite{slay2007lessons}.  
More importantly, if faults and cyberattacks are wrongly classified, they may lead to incorrect remedial actions and eventually cause severe disruptions. 

\subsection{Literature review}

 Even though detection and isolation of both faults and cyberattacks have been a field of active research over the last decade, efforts to distinguish them has remained under-explored.
In a distributed sensor network, Hidden Markov Models (HMM) have been used to distinguish between faulty and malicious data \cite{basile2006approach}. On the other hand, in \cite{rahman2014cyber} cyberattack which maliciously trip relays to disrupt power distribution has been distinguished from faults by observing the flow of fault current in the power grid. The first effort towards formalizing attack policies began with the introduction of an attack-space representation  with respect to adversary's system knowledge, disclosure and disruption resources \cite{teixeira2015secure}. This work provides replay, zero-dynamics and bias-injection injection attack policies. They also present stealthy bias-injection attack policy under incomplete system knowledge. Under a multi-agent scenario, \cite{li2019anomaly} proposes a $H_{\infty}$ optimization based observer design that distinguishes between in-domain fault and false-data injection attacks to the sensor measurement. The formulation considered in this work is restrictive in the sense that cyberattacks only affect pair-wise agents while the faults affects all the agents in the system. In contrast, \cite{rahman2016multi} uses both physical and cyber properties of a mutli-agent system (specifically smart grid) in order to achieve the same. Additionally, data-driven strategies to distinguish between fault and cyberattack has been tackled in the context of smart grids in \cite{patil2019machine,farajzadeh2021adversarial}  and for smart buildings in  \cite{anwar2015data, tertytchny2020classifying}. Lastly, \cite{bernieri2017monitoring} utilizes both model-based detection strategies along with information technology solutions to achieve the same.

\subsection{Research gap and contribution}

Literature in fault diagnostics and cyber-security reveal that a mathematical framework for distinguishing fault and cyberattack for linear system has not been proposed, to the best of our knowledge. 
Thus, to address this gap we use a sliding mode observer to estimate an anomalous input to the system and provide a criterion to distinguish whether the anomalous input is a fault or a cyberattack. 

\subsection{Organization of the chapter}
The rest of the chapter is organized as follows: Section 6.2 describes the problem set-up, Section 6.3 presents the distinguishability criterion,  Section 6.4 shows the validation of our framework through simulation studies for fault and cyberattack scenarios and finally in Section 6.6 we present the concluding remarks.

\textbf{Notations:} The  following  notations  has  been  used  in  this work: 
$\mathbf{I}_n$ is an identity matrix of size $n$, $\mathcal{R}(M)$ represents the range of matrix $M$, $B^\dagger$ represents the generalized inverse of matrix $B$, $\|\eta\|$ represents the Euclidean norm of the vector $\eta$, $\|P\|_F$ represents the  Frobenius norm of a matrix $P$.

\section{Problem Set-up}

\subsection{Cyber-physical system model}


Cyber-physical systems comprise of 6 layers: physical layer, control layer, communication layer, network layer, supervisory layer and management layer \cite{zhu2011basar}. The physical layer here represents the physical plant, sensors and actuators.
The network and communication layers contains the ICT and provides interconnections between the  physical layer, control layer and supervisory-management layers. Now, the control layer consists of the Control Module (CM) which contains controllers and state estimators. 
In contrast, the supervisory layer and management layer consists of the Central Management System and has the following role: (i) to provide  high-level supervisory management in terms of operating condition commands to the CM, (ii)  to diagnose the integrity of the plant operation using a Diagnostic Filter (DF) and (iii) to distinguish between physical faults and cyberattacks utilizing a Distinguisher Module.

\noindent
\newline
\textbf{Physical plant}
\vspace{0.1in}

Subsequently, let us consider the following linear time-invariant state space model for the physical plant $\mathcal{S}$
\begin{align}\label{sys_eta}
      \mathcal{S}: \quad 
    &\dot{x}= Ax+Bu+\eta,
    &y=x,
\end{align}
where $x\in\mathbb{R}^n$ represents the states of the system; $u\in\mathbb{R}^p$ represents the control input obtained from CM;  $A\in\mathbb{R}^{n\times n}$ represents the state matrix; \mbox{$B\in\mathbb{R}^{n\times p}$} is the actuation distribution matrix;  $\eta(t) \in \mathbb{R}^n$ represents any unknown input. In this work, we have assumed full state feedback and thus the measurement $y=x$.

Notably, under a cyber-physical setting, the adversary is able to manipulate the actuation channel of the system to launch a cyberattack $\alpha\in\mathbb{R}^p$. Consequently if a cyberattack is inflicted upon the system, the unknown input $\eta = B\alpha$. On the other hand, if there is a fault $f\in \mathbb{R}^{\Tilde{q}}$ in the system, then the unknown input $\eta = Ef$, where $E\in \mathbb{R}^{n\times\Tilde{q}}$ represents the fault distribution matrix. It is to be noted that such distribution matrix $E$ can be reliably obtained using Failure Mode and Effect Analysis (FMEA) and strategies for uncertainty quantification \cite{fmea}. We present the schematic of our problem framework in Fig.~\ref{fig:CPS_FvsC}.

 \begin{figure}[h!]
    \centering
    \includegraphics[width=0.3\textwidth]{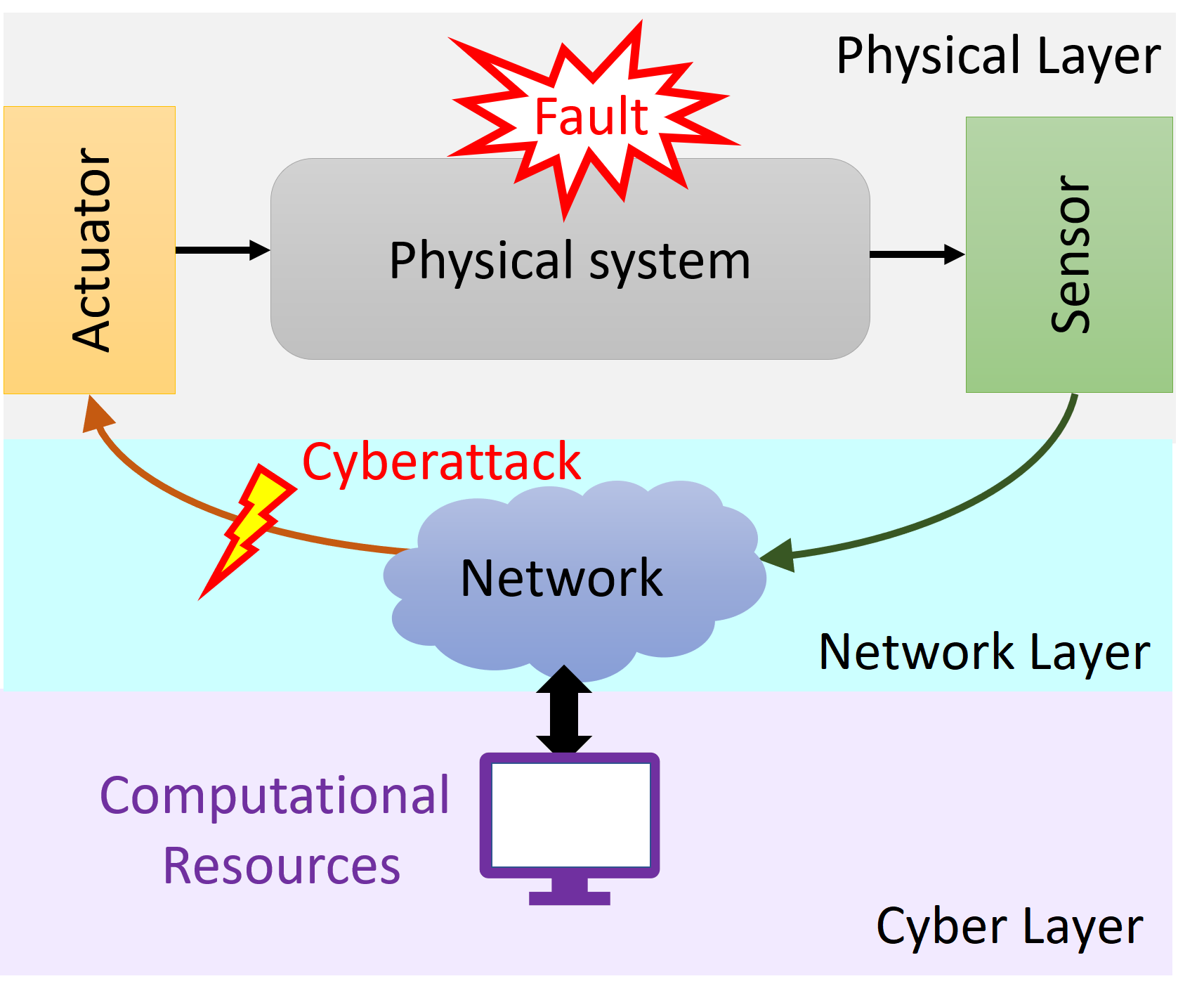}
    \caption{ The schematic diagram of the problem framework}
    \label{fig:CPS_FvsC}
\end{figure}

\begin{assm}\label{single}
In this framework, we assume that these unknown inputs are injected to the system either as cyberattacks or faults but never both simultaneously. Moreover, we assume here that the sensor measurement is not concurrently corrupted by an adversary. Such assumptions ensure that the injected cyberattack is not covert, which can evade detection \cite{teixeira2015secure}. We also assume here that the unknown input is bounded such that
\begin{align}\label{eta_max}
    \|\eta(t)\|<M<\infty, \quad  \text{for} \quad t\in [0, T_{max}], T_{max}<\infty.
\end{align}
\end{assm}
\noindent Under Assumption \ref{single}, the \textbf{system under fault} becomes:
\begin{align}\label{sys_fault}
    \mathcal{S}:\quad&\dot{x}= Ax+Bu+Ef,\quad y=x.
\end{align}
Similarly, \textbf{system under cyberattack} becomes:
\begin{align}\label{sys}
    \mathcal{S}: \quad
    &\dot{x}= Ax+B(u+\alpha),\quad y=x.
\end{align}

In the next section, we propose a sliding mode-based diagnostic filter that detects and estimates unknown inputs defined in Eqn.~\eqref{sys_eta}.

\noindent
\newline
\textbf{Diagnostic Filter (DF)}
\vspace{0.1in}

In this formulation, the objective of the sliding mode-based Diagnostic Filter (DF) is to detect and estimate unknown inputs (such as cyberattacks and faults) to the system. The structure of filter considered here is based on measurement feedback from system \cite{edwards2000sliding, utkin2017sliding} and the unknown input $\eta$ is estimated using an equivalent output error injection term.

Let us first present structure of the sliding mode-based DF as
\begin{align}
    \label{SMDF_eq}
    \mathcal{DF}: \quad 
    &\dot{\hat{x}}= A\hat{x}+Bu+ L \frac{y-\hat{y}}{\|y-\hat{y}\|},
    &\hat{y}=\hat{x},
\end{align}
where  $L\in \mathbb{R}^{n\times n}$ is the filter gain. Next, let us know define the error state as $e := x-\hat{x}$. The error dynamics is then given by
\begin{align}\label{error_full}
    &\dot{e}= Ae +\eta - L \frac{x-\hat{x}}{\|x-\hat{x}\|}.
\end{align}

\begin{propp}[Convergence of sliding mode-based Diagnostic Filter]\label{SMDF}
Consider the system given by Eqn.~\eqref{sys_eta} and the sliding mode-based DF given by Eqn.~\eqref{SMDF_eq}. If there exists positive definite matrices $P\in\mathbb{R}^{n\times n}$ and $Q\in\mathbb{R}^{n\times n}$, and constant $\gamma>0$ such that 
\begin{align}\label{P}
   A^TP+PA \leq -Q, \quad \text{where }  \text{and} \quad \gamma > \|P\|_FM,
\end{align}
where $M$ is obtained from Eqn.~\eqref{eta_max} and we choose filter gain $L$ such that
    $L =  \gamma P^{-1},$
then the estimate for the unknown input vector $\eta$ is given by 
\begin{align}\label{eta_est}
    \hat{\eta} = \mathcal{F}\left(L\frac{y-\hat{y}}{\|y-\hat{y}\|}\right),
\end{align}
where $\mathcal{F}(.)$ is a low pass filter function and $\eta\to \hat{\eta}$ in finite time.
\end{propp}

\begin{proof}
Let us define a Lyapunov functional
    $V(e) =  e^T P e,$
where $P$ is determined by Eqn.~\eqref{P}. Taking time derivative of $V(e)$ and using Eqn.~\eqref{error_full}, we obtain
\begin{align}
    \dot{V} = e^T(A^TP+PA)e + 2e^TP\eta - 2e^TPL \frac{e}{\|e\|}.
\end{align}
Now choosing filter gain $L=\gamma P^{-1}$ and using Eqn.~\eqref{P}, we obtain
\begin{align}
    \dot{V} \leqslant -e^TQe+ 2\|e\|\|P\|_FM- 2\gamma \|e\|.\label{eqxx}
\end{align}
Since $-e^TQe<0$ due to positive definiteness of $Q$, we can write Eqn.~\eqref{eqxx} as $
    \dot{V} \leqslant 2\|e\|\|P\|_FM- 2\gamma \|e\|.\label{eqxx2}
$
Considering the fact that $\gamma > \|P\|_FM$, we can write $\dot{V} \leqslant -2\beta_{+}\|e\|$ where $\beta_{+}=(-\|P\|_FM+\gamma)>0$. From this, and considering $V \geq \lambda_{min}(P)\|e\|^2 \implies \|e\| \leq \sqrt{V/\lambda_{min}(P)}$, we can further write $\dot{V} \leqslant -\beta \sqrt{V}$ where $\beta = 2\frac{\beta_{+}}{\sqrt{\lambda_{min}(P)}}>0$. This implies that $e\to 0$ as $t\to T_{max}$ where $T_{max}<\infty$ is a finite time \cite{utkin2017sliding, bhat2000finite}. 

Consequently, after $t>T_{max}$, Eqn.~\eqref{error_full} becomes 
\begin{align}\label{error_full22}
    &0= \eta - L \frac{x-\hat{x}}{\|x-\hat{x}\|},
\end{align}
and yields the equivalent output error dynamics \cite{edwards2000sliding, utkin2017sliding}. Subsequently, we can use Eqn.~\eqref{error_full22} to obtain an estimate of the unknown input by passing the output error injection term through a low pass filter $\mathcal{F}(.)$ \cite{edwards2000sliding, utkin2017sliding}. In this case, the low pass filter is chosen as
$
    \dot{\hat{\eta}} = -\frac{1}{\tau} \hat{\eta} +\frac{1}{\tau}\left(L\frac{y-\hat{y}}{\|y-\hat{y}\|}\right),
$
where $\tau$ is the filtering time constant.
\end{proof}


\section{Distinguishability}


Using the sliding mode-based DF from Proposition~\ref{SMDF}, we use the estimation of the unknown input $\eta(t)$ for distinguishing between cyberattack and fault. Now, if $\eta$ is a cyberattack, then $\eta=B\alpha$ i.e. $\mathcal{R}(B)$ represents the plausible set of cyberattacks. Similarly, $\mathcal{R}(E)$ represents the  plausible set of faults. Therefore, the question of distinguishability translates to identifying if $\eta$ lies in $\mathcal{R}(B)$ or $\mathcal{R}(E)$. With this intent, we define the following \textit{Distinguishability metric}.

\begin{deff}[Distinguishability metric]
For an unknown input $\eta\not\equiv 0$, we define a Distinguishability metric given by
functional: 
\begin{align}\label{M_x}\nonumber
    \mathcal{M}(\eta) :=&
    \Big\| \big(\mathbf{I}_n - B(B^TB)^{-1}B^T \big)\eta\Big\|_2 \\
   &\hspace{2em} - \Big\| \big(\mathbf{I}_n - E(E^TE)^{-1}E^T\big)\eta\Big\|_2.
\end{align}
\end{deff}


In the above functional definition, $\Big\| \big(\mathbf{I}_n - B(B^TB)^{-1}B^T \big)\eta\Big\|_2$ denotes the minimum distance of $\eta(t)$ from $\mathcal{R}(B)$, while $\Big\| \big(\mathbf{I}_n - E(E^TE)^{-1}E^T\big)\eta\Big\|_2$ denotes the minimum distance of $\eta(t)$ from $\mathcal{R}(E)$. Hence, $\mathcal{M}$ represents how closer or further $\eta$ is to the range space of $B$ or $E$. Evidently, positive $\mathcal{M}$ implies that the distance of $\eta$ to $\mathcal{R}(B)$ is more than the distance of $\eta$ to $\mathcal{R}(E)$. Similarly, if $\mathcal{M}$ is negative, it implies $\eta$ is closer to the $\mathcal{R}(B)$. Thus $\mathcal{M}=0$ for a non-zero $\eta$ implies that the unknown input is equidistant from both $\mathcal{R}(B)$ and $\mathcal{R}(E)$. 

\begin{figure}[!t]
\centering
\subfloat[]{\includegraphics[width=0.25\textwidth]{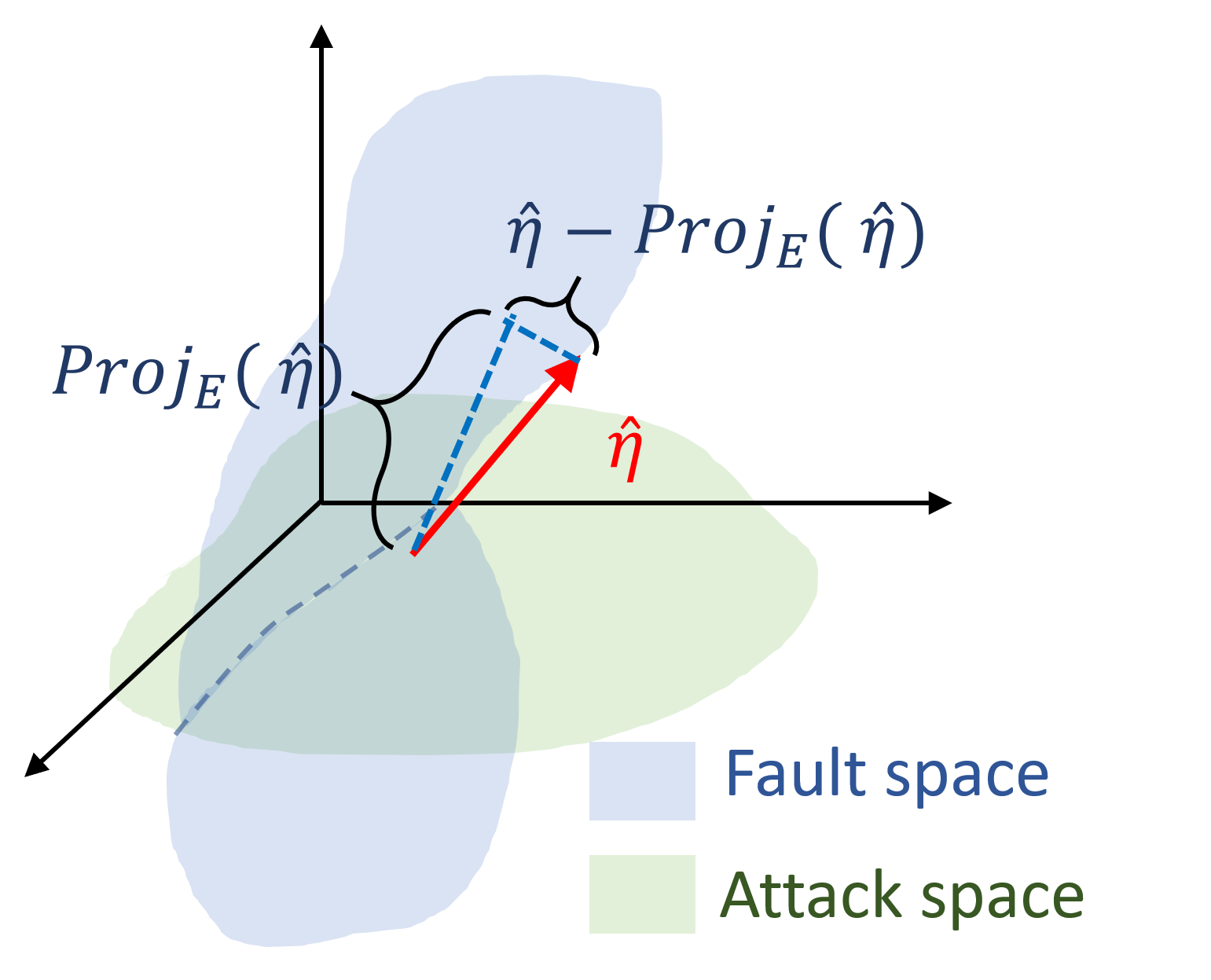}%
\label{fig:f_space}}
\subfloat[]{\includegraphics[width=0.25\textwidth]{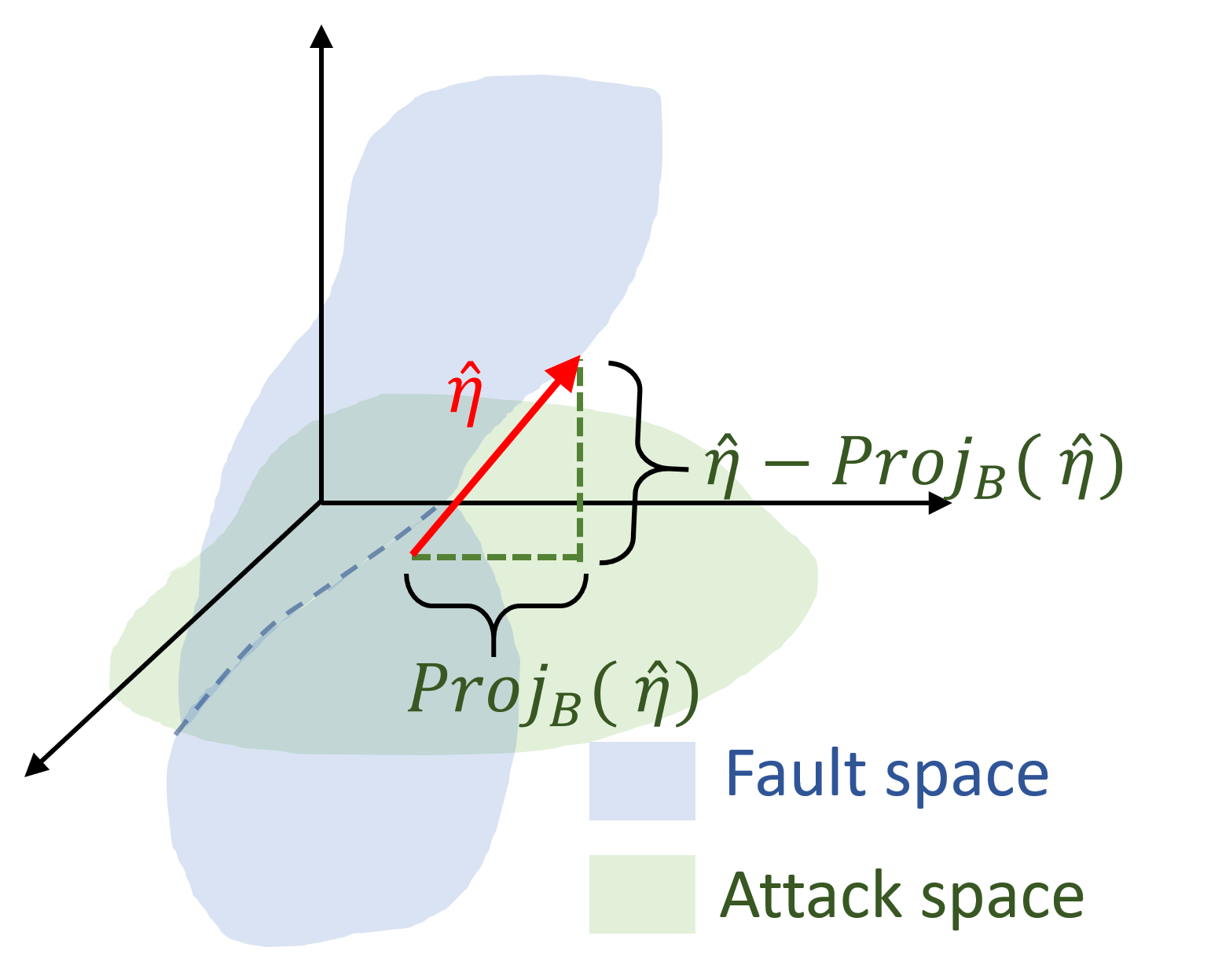}%
\label{fig:a_space}}
\caption{{Realization of Distinguishability Metric $\mathcal{M}$ using the distance from (a) fault space and (b) attack space}}
\label{fig:fault_attack_space}
\end{figure}

\begin{remm}\label{remm_proof}
We note here that for an arbitrarily accurate estimation of the unknown input $\eta$, we will have only two scenarios. Either there exists an $\alpha$ such that $\eta -B\alpha \equiv 0$ or an $f$ such that $\eta_1 - Ef \equiv 0$. This would have unambiguously proven that the unknown input $\eta$ is in fact a cyberattack in the first case and a fault for the second. However, such arbitrarily accurate estimation of unknown inputs is unrealistic not only from the point of detector design but also the presence of uncertainties in system model and measurements. Hence, function Eqn.~\eqref{M_x} is defined to obtain the degree of closeness of the unknown input to the space of plausible cyberattacks and faults. Fig.~\ref{fig:fault_attack_space} presents the geometric interpretation of the Distinguishability metric $\mathcal{M}$.
\end{remm}


The next theorem provides us with the  Distinguishability criterion .

\begin{thmm}[Distinguishability criterion]\label{DC_thm}
Let us consider the system Eqn.~\eqref{sys_eta} with non-zero unknown input $\eta(t)$ and sliding mode-based DF given by Eqn.~\eqref{SMDF_eq}. Let us also assume that the DF satisfy conditions provided in Proposition \ref{SMDF}. Then this estimated unknown input Eqn.~\eqref{eta_est} is distinguishable as either cyberattack or fault if  and only if
\begin{align}
    \mathcal{M}(\hat{\eta}) <0 & \implies \text{ cyberattack},\\
    \mathcal{M}(\hat{\eta}) >0 & \implies \text{ fault}.
\end{align}
Inversely, unknown inputs are  \textit{indistinguishable} as a fault or cyberattack if and only if $
   \mathcal{M}(\hat{\eta})= 0.$
\end{thmm}

\begin{proof}
The proof follows from the definition in Eqn.~\eqref{M_x}

\end{proof}

\begin{Cor}[Fault-mimicking cyberattacks]
A cyberattack will be indistinguishable from a fault if there exists an $\alpha_{\star}\in \mathbb{R}^p$ such that $
    \alpha_{\star}= B^{\dagger}Ef,\, \text{ for some } f\in \mathbb{R}^{\Tilde{p}},
$
and such cyberattacks are called  fault-mimicking cyberattacks. 
\end{Cor}



\section{Simulation results}
In this section, we illustrate the proposed concepts using simulation studies. The system considered is given as follows: $
    A = \begin{bmatrix}
    -30 & 0\\ 0 & -20
    \end{bmatrix}, \, B = \begin{bmatrix} 3 \\ 2\end{bmatrix}, \, D = 0,\, \text{ and } E = \begin{bmatrix} 2 \\ 5\end{bmatrix}.
$
The unknown input is $\eta = [\eta_1, \eta_2]^T.$ For all the case studies in this section, a dynamic input profile is given as input $u$ to the system Eqn.~\eqref{sys_eta} (shown in Fig.~\ref{fig:input}).The observer gain for the sliding mode-based DF Eqn.~\eqref{SMDF_eq} is chosen to be $L=50\mathbf{I}_2$ and the filter gain as $\tau=0.1$. 
\begin{figure}[h!]
    \centering
    \includegraphics[width=0.4\textwidth]{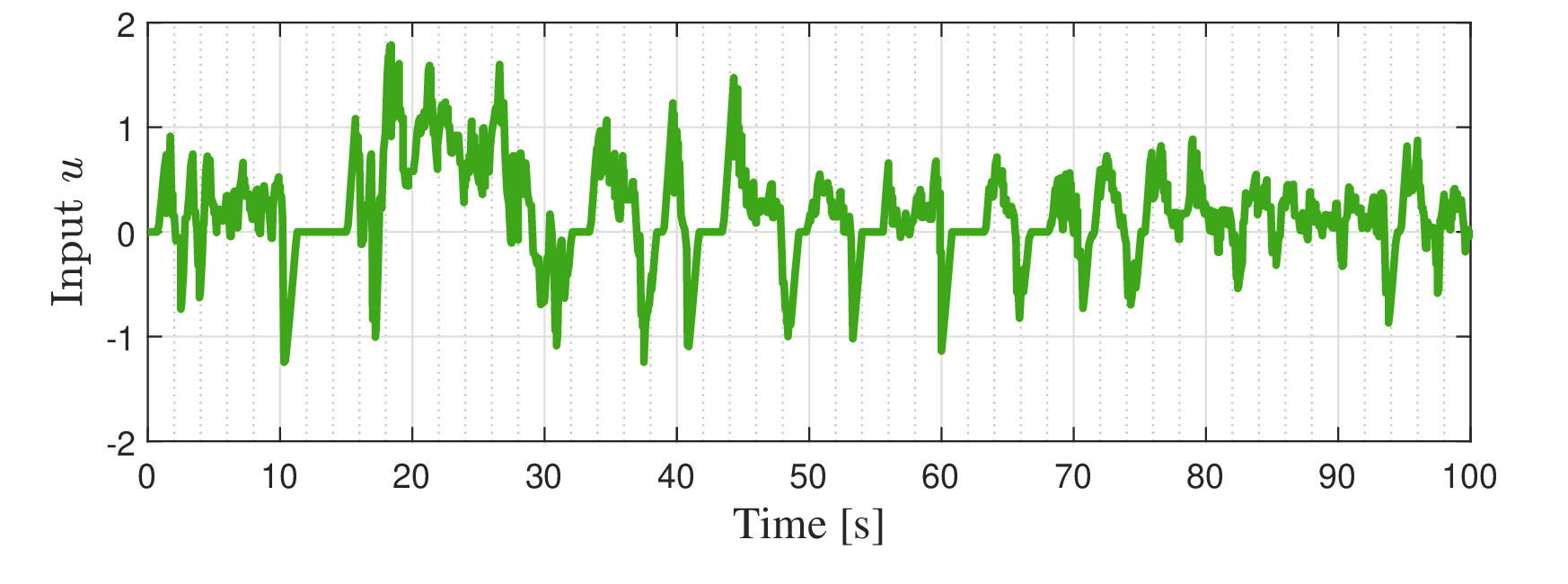}
    \caption{Dynamic input signal $u$ to the system Eqn.~\eqref{sys_eta}}
    \label{fig:input}
\end{figure}

In this simulation study we consider two cases. For Case 1, the system is under a physical fault. While for Case 2, the system is subjected to a cyberattack. The objective of this case study is to show how the unknown input estimated by the sliding mode-based DF can be successfully identified as either fault or cyberattack using the Distinguishability metric and criterion proposed in Eqn.~\eqref{M_x} and Theorem \ref{DC_thm}. With this, let us look at the results of the two case studies.

\subsection{Case 1: Fault}

For this case, the system is subject to a fault of magnitude $f= 5(1-exp(-10^{-4}t))$ and it is manifested to the system as an unknown input $\eta=Ef$. Fig.~\ref{fig:est_fault} shows that the sliding mode-based DF can estimate this unknown input while starting from arbitrary initial conditions (shown in inset). The estimates of the two components of $\hat{\eta}=[\hat{\eta}_1, \hat{\eta}2]^T$ matches with the true components of the unknown input $\eta=[\eta_1, \eta_2]^T$ and is shown in Fig.~\ref{fig:est_fault}.
\begin{figure}[h!]
    \centering
    \includegraphics[width=0.5\textwidth]{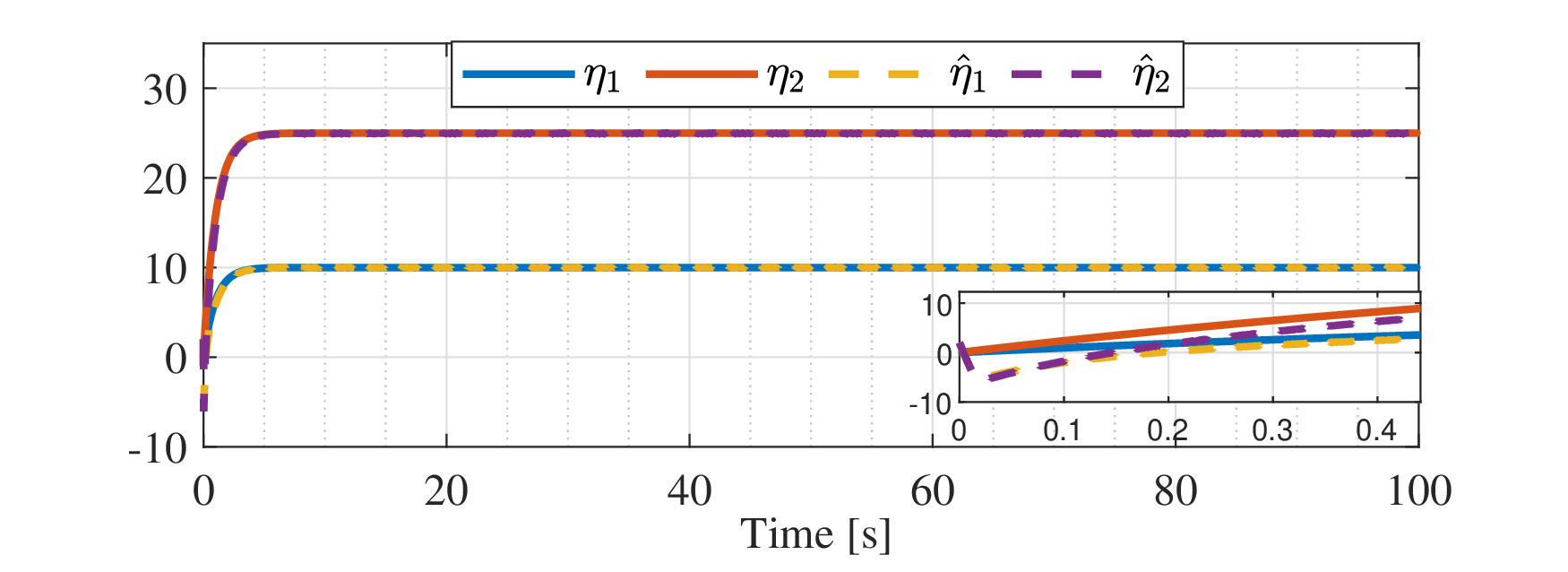}
    \caption{Fault estimated by sliding mode-based DF}
    \label{fig:est_fault}
\end{figure}

The DF also estimates the states of system $\hat{x}_1$ and $\hat{x}_2$ and these estimates matches significantly with the true states $x_1$ and $x_2$. 
We observe in Fig.~\ref{fig:metric_fault}  that $\mathcal{M}>0$  in steady state, indicating that the unknown input is a fault. We also note here that the Distinguishability metric is non-positive for the first 0.2s (as seen in the inset of Fig.~\ref{fig:metric_fault}). However, this is due to the time needed by the DF in order to converge to the correct estimates of the states and unknown inputs to the system.




\begin{figure}[h!]
    \centering
    \includegraphics[width=0.45\textwidth]{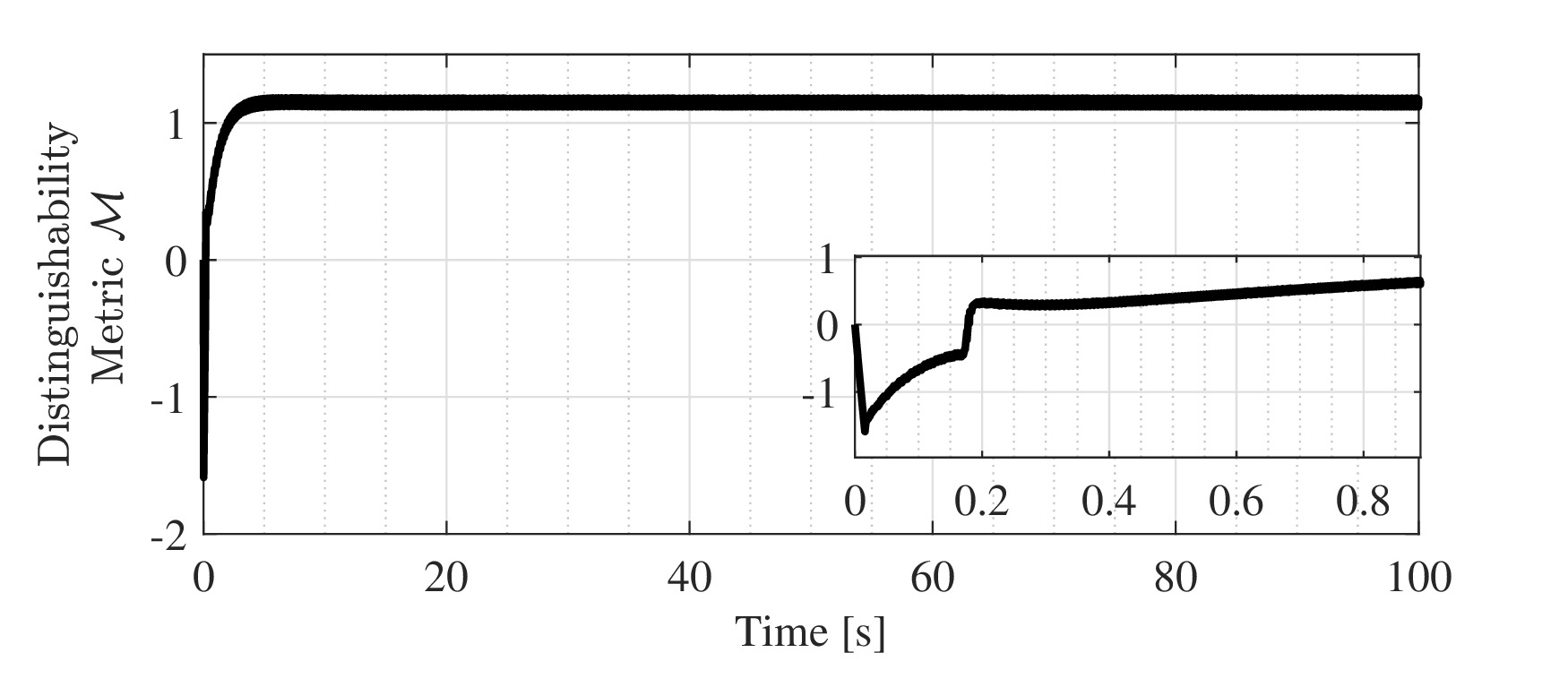}
    \caption{Positive Distinguishability metric $\mathcal{M}$ denotes the estimated unknown input as fault}
    \label{fig:metric_fault}
\end{figure}

\subsection{Case 2: Cyberattack}
For this case study, we construct a cyberattack $\alpha$ that can drive the system states to some unintended states $[1,1]$. Using steady state condition of the system equation Eqn.~\eqref{sys_eta}, we obtain the attack policy to be
$
    \alpha = -B^{\dagger}A[1,1]^T-u.
$
From Fig.~\ref{fig:est_attack}, it is evident that the sliding mode-based DF can faithfully estimate the two components of the unknown input $\eta = [\eta_1, \eta_2]^T=B\alpha.$ The initial condition for the unknown inputs are unspecified. Hence, the estimates are initialized arbitrarily. However, the DF converges to correct estimates in the steady state starting from the arbitrary initial conditions (as seen in the inset of Fig.~\ref{fig:est_attack}).  
Subsequently, we calculate the Distinguishability Metric $\mathcal{M}$ from Eqn.~\eqref{M_x} and plot in Fig.~\ref{fig:metric_attack}. Since $\mathcal{M}<0$, we can conclude from the Distinguishability criterion that the unknown input $\eta$ is a cyberattack.

\begin{figure}[h!]
    \centering
    \includegraphics[width=0.5\textwidth]{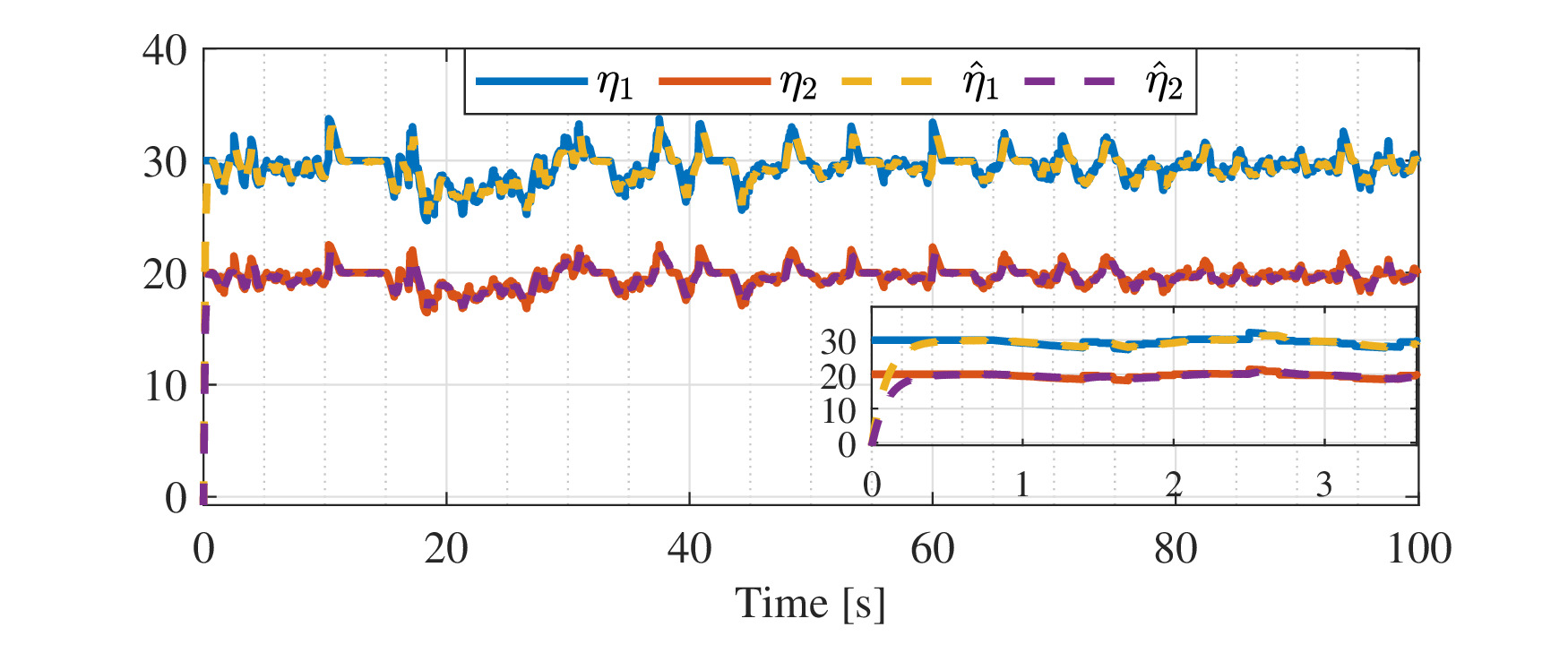}
    \caption{Cyberattack estimated by sliding mode-based DF}
    \label{fig:est_attack}
\end{figure}\vspace*{-1cm}
\begin{figure}[h!]
    \centering
    \includegraphics[width=0.45\textwidth]{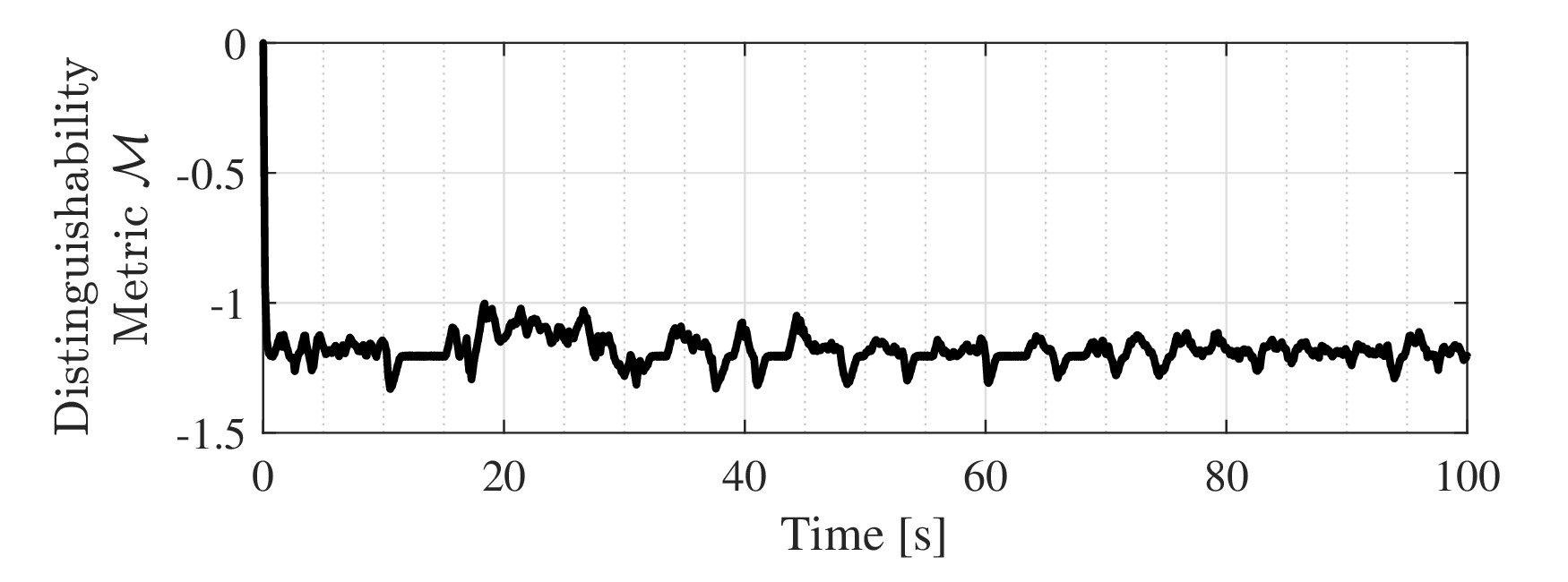}
    \caption{Negative Distinguishability metric $\mathcal{M}$ denotes the estimated unknown input as cyberattack}
    \label{fig:metric_attack}
\end{figure}

\section{Conclusion}
Distinguishing between the occurrence of faults and cyberattacks in system are of utmost importance in order to provide appropriate mitigation strategies. In this work, we have proposed a mathematical framework for distinguishing between the two, utilizing a sliding mode-based Diagnostic Filter (DF). Using this estimate for the unknown input, we proposed a Distinguishability metric and criterion in order to achieve our goal. Finally, we have conducted a set of simulation studies with both faults and cyberattacks to illustrate the validity of our proposed framework.

\bibliographystyle{asmems4}

\bibliography{ref.bib}



\end{document}